\documentclass[letterpaper, 10 pt, conference]{ieeeconf}  

\IEEEoverridecommandlockouts                              %
\let\labelindent\relax                                                         

\usepackage{amsthm}
\newif\ifproof
\prooftrue

\usepackage{amsmath,amssymb}
\usepackage{xcolor}
\usepackage{subcaption}
\usepackage{booktabs} \usepackage{tikz}
\usetikzlibrary{automata,positioning}
\usetikzlibrary{shapes,arrows}
\tikzset{state/.style={circle, draw, minimum size=0.5cm, initial distance=0.2cm}}

\usepackage{todonotes}
\usepackage[noadjust]{cite}

\usepackage{algorithm}
\usepackage{enumitem}
\usepackage{algorithmicx}
\usepackage{algpseudocode}
\algnewcommand{\LineComment}[1]{\Statex \(\triangleright\) #1}
\makeatletter
\newcommand{\algmargin}{\the\ALG@thistlm}
\makeatletter
\newcommand*{\rom}[1]{\expandafter\@slowromancap\romannumeral #1@}
\makeatother
\newlength{\whilewidth}
\algdef{SE}[parWHILE]{parWhile}{EndparWhile}[1]
  {\parbox[t]{\dimexpr\linewidth-\algmargin}{     \hangindent\whilewidth\strut\algorithmicwhile\ #1\ \algorithmicdo\strut}}{\algorithmicend\ \algorithmicwhile}\algnewcommand{\parState}[1]{\State  \parbox[t]{\dimexpr\linewidth-\algmargin}{\strut #1\strut}}
\algnewcommand{\parRequire}[1]{\Require  \parbox[t]{\dimexpr\linewidth-\algmargin}{\strut #1\strut}}
\DeclareMathAlphabet{\mathpzc}{OT1}{pzc}{m}{it}
\usepackage[hidelinks]{hyperref}
\usepackage{cleveref}
\usepackage{etoolbox}
\usepackage{stackengine}
\def\delequal{\mathrel{\ensurestackMath{\stackon[1pt]{=}{\scriptstyle\Delta}}}}

\newtheorem{definition}{Definition}
\newtheorem{theorem}{Theorem}
\newtheorem{remark}{Remark}

\newtheorem{lemma}{Lemma}

\newtheorem{corollary}{Corollary}
\newtheorem{example}{Example}

\crefname{section}{Section}{Sections}
\crefname{subsection}{Section}{Sections}
\crefname{definition}{Definition}{Definitions}
\crefname{proposition}{Proposition}{Propositions}
\crefname{lemma}{Lemma}{Lemmas}
\crefname{theorem}{Theorem}{Theorems}
\crefname{corollary}{Corollary}{Corollaries}
\crefname{example}{Example}{Examples}
\crefname{figure}{Figure}{Figures}
\crefname{assumption}{Assumption}{Assumptions}
\crefname{remark}{Remark}{Remarks}
\crefname{running}{Running Example}{Running Examples}
\crefname{algorithm}{Algorithm}{Algorithms}

\newcommand\old[1]{{\color{gray} #1}}

\renewcommand\old[1]{}

\title{\LARGE \bf
Resiliency of Nonlinear Control Systems to Stealthy Sensor Attacks
}

\author{Amir Khazraei and Miroslav Pajic
\thanks{The authors are with the Department of Electrical \& Computer Engineering, Duke University, Durham, NC 27708. Email: {\tt\small \{amir.khazraei, miroslav.pajic\}@duke.edu.}}
\thanks{This work is sponsored in part by the ONR under agreement N00014-20-1-2745, AFOSR under the award number FA9550-19-1-0169, as well as by the NSF under CNS-1652544 award and the National AI Institute for Edge Computing Leveraging Next Generation Wireless Networks, Grant CNS-2112562.}
}

\begin{document}

\maketitle
\thispagestyle{empty}
\pagestyle{empty}

\begin{abstract}
In this work, we focus on analyzing vulnerability of nonlinear dynamical control systems to stealthy sensor attacks. We start by defining the notion of stealthy attacks in the most general form by leveraging Neyman-Pearson lemma; specifically, an attack is considered to be stealthy if it is stealthy from (i.e., undetected by) any intrusion detector -- i.e., the probability of the detection is not better than a random  guess.
We then provide a sufficient condition under which a nonlinear control system is vulnerable to stealthy attacks, in terms of moving the system to an unsafe region due to the attacks. In particular, we show that if the closed-loop system is incrementally exponentially stable while the open-loop 
plant is incrementally unstable, then the 
system is 
vulnerable to stealthy yet impactful attacks on sensors. Finally, we illustrate our results on a case study. 
\end{abstract}

\section{Introduction}  
\label{sec:intro}

Cyber-physical systems (CPS) are characterized by the tight integration of controllers and physical plants, potentially through communication networks. 
As such, they 
have been shown to be 
vulnerable to various types of cyber and physical attacks 
with disastrous impact (e.g.,~\cite{chen2011lessons}). Consequently, as part of the control design and analysis process, it is critical to identify early any vulnerability of the considered system to impactful attacks, especially the ones that are potentially stealthy to the deployed intrusion detection mechanisms.

Depending on 
attacker 
capabilities, different types of stealthy attacks have been proposed. For instance, when only sensor measurements can be compromised by the attacker, it has been shown that false data injection attacks are capable of significantly impacting the system while remaining undetected (i.e., stealthy) by a particular type of residual-based anomaly detectors (e.g.,~\cite{mo2010false,jovanov_tac19,kwon2014analysis,khazraei_automatica21,khazraei_acc20,zhang2020false,shang2021optimal}). 
For example, for linear time invariant (LTI) systems, if measurements from all sensors can be compromised , the plant's (i.e., open-loop) instability 
is a necessary and sufficient condition for the existence of impactful stealthy attacks. Similarly, for LTI systems with strictly proper transfer functions, the attacker that  compromises the control input can design effective  stealthy attacks if the system has unstable zero invariant (e.g.,~\cite{teixeira2012revealing,pasqualetti2013attack}); however, when the transfer function is not strictly proper, the attacker needs to compromise both plant's inputs and outputs. 
When the attacker compromises both the plant's actuation and sensing, 
e.g., ~\cite{sui2020vulnerability} derives the conditions under which the system is vulnerable to stealthy~attacks. 


However, the common assumption for all these results is that the considered plant is an LTI system. Furthermore, the notion of stealthiness is only characterized  for a \emph{specific type} of the employed intrusion detector (e.g., $\chi^2$-based detectors). In~\cite{bai2017kalman,bai2017data},  the notion of attack stealthiness is generalized, defining an attack as stealthy if it is stealthy from the best existing intrusion detector. In addition, the authors show that a sufficient condition for such notion of stealthiness is that the Kullback–Leibler (KL) divergence between the probability distribution of compromised system measurements and the attack-free measurements is close to zero, and consider stealthiness of such attacks on control systems with an LTI plant and an LQG controller. 


To the best of our knowledge, no existing work provides 
vulnerability analysis for systems with nonlinear dynamics, while considering general control and intrusion detector designs. 
In~\cite{smith2015covert}, covert attacks are introduced as stealthy attacks that can target a potentially nonlinear system. 
However, the attacker needs to have perfect knowledge of the system's dynamics and be able  to compromise \emph{both} the plant's input and outputs. Even more importantly, as the attack design is based on attacks on LTI systems, no guarantees are provided for effectiveness and stealthiness of attacks on nonlinear systems.
More recently, \cite{zhang2021stealthy} introduced stealthy attacks on a \emph{specific class} of nonlinear systems with residual-based intrusion detector, but provided effective attacks only when \emph{both} plant's inputs and outputs are compromised by the attacker. On the other hand, in this work, we assume the attacker can only compromise the  plant's sensing data and  consider systems with \emph{general} nonlinear dynamics. For systems with general nonlinear dynamics and residual-based intrusion detectors, machine learning-based methods to design the stealthy attacks have been introduced (e.g.,~\cite{khazraei2021learning}), but without any theoretical analysis and guarantees regarding the impact of the stealthy attacks. 

Consequently, in this work we provide conditions for existence of effective yet stealthy attacks on nonlinear systems without limiting the analysis on particular type of employed intrusion detectors. 
Our notion of attack stealthiness and system performance degradation is closely related to~\cite{khazraeReport21}. However, we extend these notions for systems with general nonlinear plants and controllers. To the best of our knowledge, this is the first work that considers the problem of stealthy impactful sensor attacks for systems with general nonlinear dynamics that is independent of the deployed intrusion detector.  The main contributions of the paper are twofold. First, we introduce the notions of \emph{strict} and \emph{$\epsilon$-stealthiness}. 
Second, using the well-known results for incremental stability introduced in~\cite{angeli2002lyapunov}, we derive  conditions for the existence of effective stealthy attacks that move the system into an unsafe operating region.  We show that if the closed-loop system is incrementally stable while the open-loop plant is incrementally unstable, then the closed-loop system is strictly vulnerable to stealthy sensing attacks.

The paper is organized as follows. In~\cref{sec:prelim}, we introduce preliminaries, whereas \cref{sec:motive} presents the system and attack model, before formalizing the notion of stealthiness in \cref{sec:stealthy}. 
\cref{sec:perfect} provides sufficient conditions for existence of the impactful yet stealthy attacks.
Finally, in  \cref{sec:simulation}, we illustrate our results on a case-study, before concluding remarks in \cref{sec:concl}. 

\paragraph*{Notation}
We use $\mathbb{R, Z}, \mathbb{Z}_{t\geq 0}$ to denote  the sets of reals, integers and non-negative integers, respectively, and $\mathbb{P}$ denotes the probability for a random variable.  For 
a square matrix $A$, $\lambda_{max}(A)$ denotes the maximum eigenvalue. 
For a vector $x\in{\mathbb{R}^n}$, $||x||_p$ denotes the $p$-norm of $x$; when $p$ is not specified, the 2-norm is implied. 
For a vector sequence, 
$x_0:x_t$ denotes the set $\{x_0,x_1,...,x_t\}$. 
A function $f:\mathbb{R}^{n}\to \mathbb{R}^{p}$ is Lipschitz with constant $L$ if for any $x,y\in \mathbb{R}^{n}$ it holds that $||f(x)-f(y)||\leq L ||x-y||$. 
%
Finally, if $\mathbf{P}$ and $\mathbf{Q}$ are probability distributions relative to Lebesgue measure with densities $\mathbf{p}$ and $\mathbf{q}$, respectively, then 
the Kullback–Leibler  (KL) divergence between $\mathbf{P}$ and $\mathbf{Q}$ is defined as
$KL(\mathbf{P},\mathbf{Q})=\int \mathbf{p}(x)\log{\frac{\mathbf{p}(x)}{\mathbf{q}(x)}}dx$.

\section{Preliminaries}\label{sec:prelim}

Let $\mathbb{X}\subseteq \mathbb{R}^n$ and $\mathbb{D}\subseteq \mathbb{R}^m$, with $0\in \mathbb{X}, \mathbb{D}$. Consider a discrete-time nonlinear system with an exogenous input, modeled in the state-space form as
\begin{equation}\label{eq:prilim}
x_{t+1}=f(x_t,d_t),\quad x_t\in \mathbb{X},\,\,t\in \mathbb{Z}_{t\geq 0},
\end{equation}
where $f:\mathbb{X}\times\mathbb{D}\to \mathbb{X}$ is continuous and $f(0,0)=0$. We denote by $x(t,\xi,d)$ 
the trajectory (i.e., the solution) of~\eqref{eq:prilim} at time $t$, when the system has the initial condition $\xi$ and is subject to the input sequence $\{d_0:d_{t-1}\}$.\footnote{To simplify our notation, we denote the sequence $\{d_0:d_{t-1}\}$ as $d$.} 

The following definitions are derived from~\cite{angeli2002lyapunov,tran2018convergence,tran2016incremental}.

\begin{definition}
The system~\eqref{eq:prilim} is incrementally exponentially stable (IES) in the set $\mathbb{X}\subseteq \mathbb{R}^n$ if there exist $\kappa>1$ and
$\lambda>1$ such that 
\begin{equation}
\Vert x(t,\xi_1,d)-x(t,\xi_2,d)\Vert \leq \kappa \Vert \xi_1-\xi_2\Vert \lambda^{-t},
\end{equation}
holds for all $\xi_1,\xi_2\in \mathbb{X}$, any $d_t\in \mathbb{D}$, and $t\in \mathbb{Z}_{t\geq 0}$. When $\mathbb{X}=\mathbb{R}^n$, the system is referred to as globally incrementally exponentially stable (GIES).
\end{definition}

\begin{definition}
The system~\eqref{eq:prilim} is incrementally  unstable (IU) in the set $\mathbb{X}\subseteq \mathbb{R}^n$ if for all $\xi_1\in \mathbb{X}$ and any $d_t\in \mathbb{D}$, there exists a $\xi_2$ such that for any $M>0$, 
\begin{equation}
\Vert x(t,\xi_1,d)-x(t,\xi_2,d)\Vert \geq M,
\end{equation}
holds for all $t\geq t'$, for some $t'\in \mathbb{Z}_{t\geq 0}$.
\end{definition}

\section{System Model} 
\label{sec:motive} 
In this section, we introduce the considered system and attack model, allowing us to formally capture the problem addressed in this work. 

\subsection{System and Attack Model} \label{sec:A} 
We consider the setup from \cref{fig:architecture} where each of the components is modeled as follows.

\subsubsection{Plant}
We assume that the states of the system evolve following a general nonlinear discrete-time dynamics that can be captured in the state-space form~as  
\begin{equation}\label{eq:plant}
\begin{split}
{x}_{t+1} &= f(x_t,u_t)+w_t,\\
y_t &= h(x_t)+v_t;
\end{split}
\end{equation}
here, $x \in {\mathbb{R}^n}$, $u \in {\mathbb{R}^m}$, $y \in {\mathbb{R}^p}$ are the state, input and output vectors of the plant, respectively. In addition, $f$ is a nonlinear mapping from previous time state and control input to the current state, and $h$ is the mapping from the states to the sensor measurements; we assume here that $h$ is Lipschitz with a constant $L_h$.  The plant output vector captures measurements from the set of plant sensors $\mathcal{S}$. 
Furthermore, $w \in {\mathbb{R}^{n}}$ and $v \in {\mathbb{R}^p}$ are the process and measurement noises that are assumed to be Gaussian with zero mean, and $\Sigma_w$ and $\Sigma_v$ covariance matrices, respectively. 

As we show later, it will be useful to consider the input to state relation of the dynamics~\eqref{eq:plant}; if we define $U=\begin{bmatrix}u^T&w^T\end{bmatrix}^T$, the first equation in~\eqref{eq:plant} becomes
\begin{equation}\label{eq:input-state}
x_{t+1}=f_u(x_t,U_t).
\end{equation}

\begin{figure}[!t]
\centering
\vspace{6pt}
\includegraphics[width=0.468\textwidth]{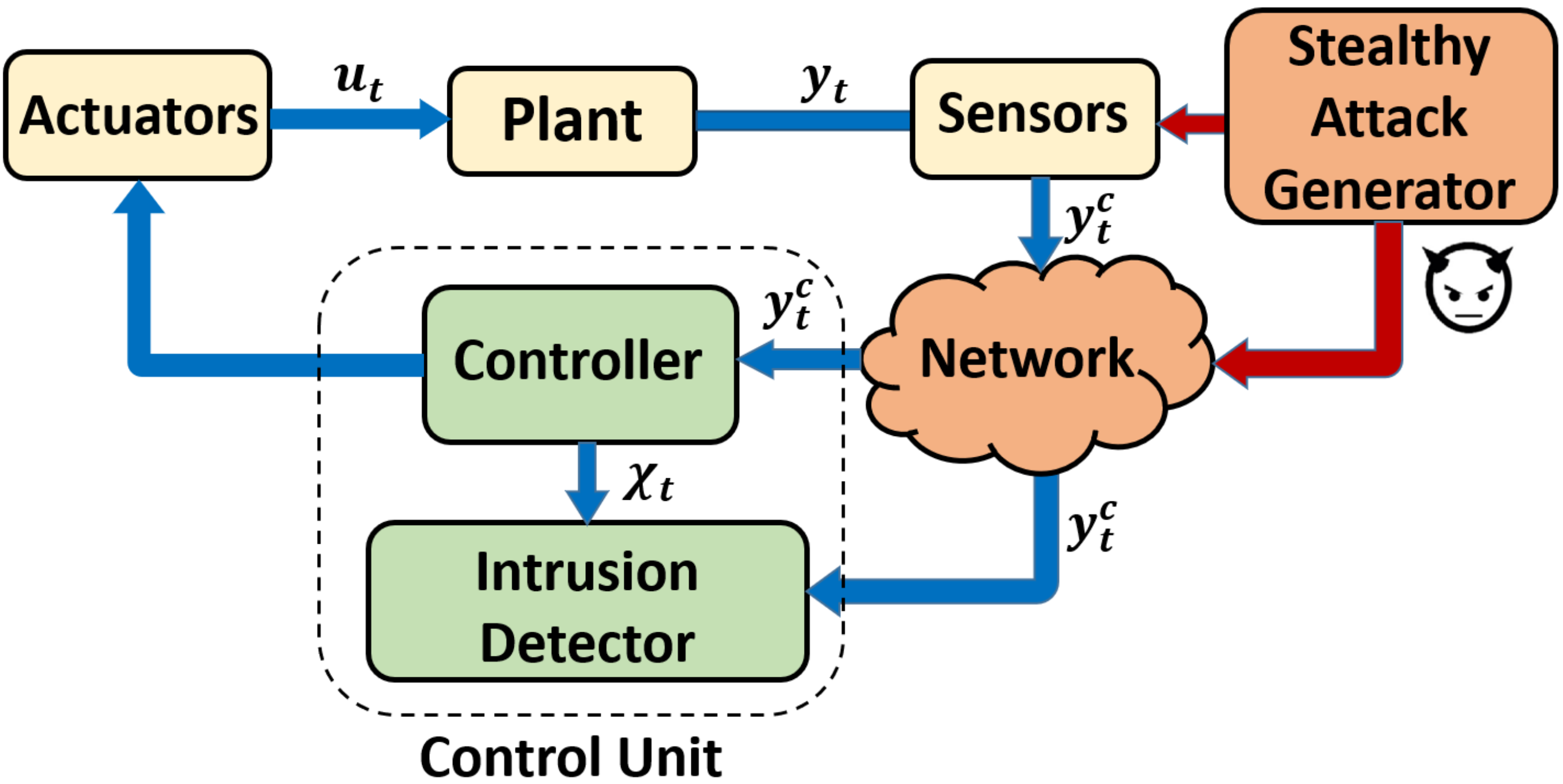}
\caption{Control system architecture considered in this work, in the presence of network-based attacks.}
\label{fig:architecture}
\end{figure}

\subsubsection{Control Unit}
The controller, illustrated in~\cref{fig:architecture}, is equipped with a feedback controller in the most general form, as well as an intrusion detector (ID). In what follows, we provide more details on the controller design. Intrusion detector will be discussed after introducing the attack model.

\paragraph*{Controller}
A large number of dynamical systems are intrinsically unstable or are designed to be unstable (e.g., if an aircraft is unstable, it is easier to change its altitude). Thus, it is critical to stabilize such systems using a proper controller. 
Due to their robustness to uncertainties, closed-loop controllers are utilized in 
most control systems. In the most general form, a feedback controller design can be captured in the state-space form~as 
\begin{equation}
\label{eq:plant_withoutB}
\begin{split}
\mathpzc{X}_{t} &= f_c(\mathpzc{X}_{t-1},y_t^c),\\
u_t  &= h_c(\mathpzc{X}_{t},y_t^c),
\end{split}
\end{equation}
where $\mathpzc{X}$ is the internal state of the controller, and $y^c$ captures the sensor measurements received by the controller. 
%
Thus, without malicious activity, it holds that $y^c=y$.\footnote{Here we assume that the employed communication network is reliable (e.g., wired).}
Note that the control model~\eqref{eq:plant_withoutB} is general, capturing for instance 
nonlinear filtering followed by a classic nonlinear controller (e.g., $f_c$ can model an extended Kalman filter and $h_c$ any full-state feedback controller).  

We define the full state of the closed-loop control system as $\mathbf{X}\delequal \begin{bmatrix}
{x}\\{\mathpzc{X}} \end{bmatrix}$, and exogenous disturbances as $\mathbf{W}\delequal \begin{bmatrix}
w\\{v} \end{bmatrix}$; then, the dynamics of the closed-loop system can be captured~as
\begin{equation}\label{eq:closed-loop}
\mathbf{X}_{t+1}=F(\mathbf{X}_t,\mathbf{W}_t).
\end{equation}

We assume that $\mathbf{X}=0$ is the operating point of the noiseless system (i.e., when $w=v=0$). {Moreover, we assume $f_c$ and $h_c$ are designed 
to keep the system within a safe region around the equilibrium point.} Here, without loss of generality, we define the safe region as $\mathbf{S}=\{x\in \mathbb{R}^n\,\,|\,\,\Vert x\Vert_2 \leq R_{\mathbf{S}}\}$, for some $R_{\mathbf{S}}>0$.

\subsubsection{Attack Model} \label{sec:attack_model}

We consider a sensor attack model where, for  sensors from the set $\mathcal{K}\subseteq{\mathcal{S}}$, 
the information delivered to the controller differs from the non-compromised sensor measurements. The attacker can achieve this via e.g., noninvasive attacks such sensor spoofing (e.g.,~\cite{kerns2014unmanned}) or by compromising information-flow from the sensors in $\mathcal{K}$ to the controller (e.g., as in network-based attacks~\cite{lesi_rtss17}). In either cases,  the attacker can launch false-date injection attacks, inserting a desired value instead of the current measurement of a compromised sensor.\footnote{We refer to sensors from $\mathcal{K}$ as compromised, even if a sensor itself is not directly
compromised but its measurements may be altered due to e.g., network-based attacks.}

Thus, assuming that the attack starts at time $t=0$, the sensor measurements delivered to the controller for $t\geq 0$ can be modeled~as~\cite{ncs_attack_model}
\begin{equation}\label{att:model}
y^{c,a}_t = y_t^a+a_t;
\end{equation}
here, $a_t\in {\mathbb{R}^p}$ denotes the attack signal injected by the attacker at time $t$ via the compromised sensors from $\mathcal{K}$, $y_t^a$ is the true sensing information (i.e., before the attack is injected at time $t$). 
In the rest of the paper we assume $\mathcal{K}=\mathcal{S}$; for some systems, we will discuss how the results can be generalized for the case when $\mathcal{K}\subset \mathcal{S}$. 

Note that since the controller uses the received sensing information to compute the input $u_t$, the compromised sensor values affect the evolution of the system and controller states. Hence, we add the superscript $a$ to denote any signal obtained from a compromised system 
-- e.g.,  thus, $y_t^a$ is used to denote before-attack sensor measurements when the system is under attack in~\eqref{att:model}, and we denote the closed-loop plant and controller state when the system is compromised as $\mathbf{X}^a\delequal \begin{bmatrix}
{x^a}\\{{\mathpzc{X}^a}} \end{bmatrix}$.

In this work, we consider the commonly adopted threat model as in majority of existing stealthy attack designs, e.g.,~\cite{mo2009secure,mo2010false,smith2015covert, khazraei_automatica21, jovanov_tac19}, where the attacker has full knowledge of the system, its dynamics and employed architecture. In addition, the attacker has the required computational power to calculate suitable attack signals to be injected,
while planning ahead as needed. 


Finally, the attacker's goal is to design an attack signal $a_t$, $t\geq 0$, such that it always remains \emph{stealthy} -- i.e., undetected by the intrusion detection system -- while \emph{maximizing control performance degradation}. The notions of \emph{stealthiness} and \emph{control performance degradation} depend on the employed control architecture, and thus will be formally defined after the controller and intrusion detection have been~introduced.


\subsubsection{Intrusion Detector} 
To detect system attacks (and anomalies), we assume that an intrusion detector (ID) is employed, analyzing the received sensor measurements and internal state of the controller. Specifically, by defining $Y\delequal \begin{bmatrix}
{y^c}\\{\mathpzc{X}} \end{bmatrix}$, 
as well as $Y^a\delequal \begin{bmatrix}
{y^{c,a}}\\{\mathpzc{X}^a} \end{bmatrix}$ when the system is under attack,
we assume that the intrusion detector has access to a sequence of values $Y_{-\infty}:Y_t$ until time $t$ 
and solves the binary hypothesis 
checking\\

$H_0$:  normal condition (the ID receives $Y_{-\infty}:Y_t$);~~

$H_1$: abnormal behaviour (receives 
$Y_{-\infty}:Y_{-1},Y_{0}^a:Y_t^a$).\footnote{Since the attack starts at $t=0$, we do not use superscript $a$ for the system evolution for $t<0$, as the trajectories of the non-compromised and compromised systems do not differ before the attack~starts.}
\\

Given a sequence of received data denoted by $\bar{Y}^t=\bar{Y}_{-\infty}:\bar{Y}_t$, it is either extracted from the distribution of the null hypothesis $H_0$, which we refer to as~$\mathbf{P}$, or from an \textbf{\emph{unknown}} distribution of the alternative hypothesis $H_1$, which we denote as~$\mathbf{Q}$.
{Note here that, for known noise profiles, the distribution $\mathbf{Q}$ is controlled by the injected attack signal.}

Defining the intrusion detector mapping as $D: \bar{Y}^t\to \{0,1\}$, two possible errors may occur. The error type ($\rom{1}$) known as \emph{false alarm}, occurs if $D(\bar{Y}^t)=1$ when $\bar{Y}^t \sim \mathbf{P}$ and error type ($\rom{2}$), also known as \emph{miss-detection}, occurs when $D(\bar{Y}^t)=0$ for $\bar{Y}^t \sim \mathbf{Q}$. Hence, we define 
the \underline{\emph{sum of conditional error probabilities}} of the intrusion detector for a given random sequence $\bar{Y}^t$, at time $t$~as
\begin{equation}
\label{eq:pe}
p_t^e=\mathbb{P}(D(\bar{Y}^t)=0|\bar{Y}^t\sim \mathbf{Q})+\mathbb{P}(D(\bar{Y}^t)=1|\bar{Y}^t \sim \mathbf{P}).
\end{equation}

Note that $p_t^e$ is not a probability measure as it can take values larger than one. However, it will be useful when we define the notion of stealthy attacks in the following~section.

\section{Formalizing Stealthy Attacks Requirements}
\label{sec:stealthy}
In this section, we capture the conditions for which an attack sequence is stealthy even from {{an optimal}} intrusion detector.  Specifically, we define an attack to be strictly stealthy if \emph{there exists no detector that can perform better than random guess between the two hypothesis}; by better we mean the true attack detection probability is higher than the false alarm probability. However, reaching such stealthiness guarantees may not be possible in general. Therefore, we define the notion of $\epsilon$-stealthiness, which as we will show later, is attainable for a large class of nonlinear systems. 

Before formally defining the notion of attack stealthiness, we introduce the following lemma. 

\begin{lemma}\label{lemma:stealthy}
Any intrusion detector $D$ 
cannot perform better than a random guess between the two hypothesis if and only if $p^e\geq 1$. 
Also, $p^e=1$ if and only if $D$ performs as well as a random guess detector. 
\end{lemma}

\begin{proof}
First, we consider the case $p^e> 1$. 
From~\eqref{eq:pe},
we have 
\begin{equation}
\begin{split}
1&< p^e=\mathbb{P}(D(\bar{Y})=0|\bar{Y}\sim \mathbf{Q})+\mathbb{P}(D(\bar{Y})=1|\bar{Y} \sim \mathbf{P})\\
&=1-\mathbb{P}(D(\bar{Y})=1|\bar{Y}\sim \mathbf{Q})+\mathbb{P}(D(\bar{Y})=1|\bar{Y} \sim \mathbf{P})\\
\end{split}
\end{equation}

Thus, $\mathbb{P}(D(\bar{Y})=1|\bar{Y}\sim \mathbf{Q})< \mathbb{P}(D(\bar{Y})=1|\bar{Y} \sim \mathbf{P})$. This means the probability of attack detection is less than the false alarm rate; therefore, $D$ is performing worse than random guess as in random guess we have $\mathbb{P}(D(\bar{Y})=1|\bar{Y}\sim \mathbf{Q})= \mathbb{P}(D(\bar{Y})=1|\bar{Y} \sim \mathbf{P})=\mathbb{P}(D(\bar{Y})=1)$ because random guess is independent of the given distribution. 
When the equality holds (i.e., $p^e=1$), it holds that $\mathbb{P}(D(\bar{Y})=1|\bar{Y}\sim \mathbf{Q})= \mathbb{P}(D(\bar{Y})=1|\bar{Y} \sim \mathbf{P})$ where the decision of the detector $D$ is independent of the distribution of $\bar{Y}$ and therefore, the detector performs as the random guess detector. 

Since the reverse of all these implications hold, 
the other (i.e., necessary) conditions of the theorem also hold.
\end{proof}
    
Now, using Lemma~\ref{lemma:stealthy}, we can define the notions of strict stealthiness and $\epsilon$-stealthiness as follows.

\begin{definition} \label{def:stealthiness}
Consider the system from~\eqref{eq:plant}. An attack sequence is 
\textbf{strictly stealthy} if there exists no detector such that the total error probability $p_t^e$ satisfies $p_t^e<1$, for any $t\in \mathbb{Z}_{\geq 0}$. 
An attack is
\textbf{$\epsilon$-stealthy} if for a given $\epsilon >0$, there exists no detector such that $p_t^e<1-\epsilon$, for any $t\in \mathbb{Z}_{\geq 0}$. 
\end{definition}

The following theorem uses Neyman-Pearson lemma to capture the condition for which the received sensor measurements 
satisfy the stealthiness condition in Definition~\ref{def:stealthiness}. 


\begin{theorem}[\cite{khazraeReport21}]
An attack sequence 
is 
\begin{itemize}
    \item 
    strictly stealthy if and only if  
    $KL\big(\mathbf{Q}(Y_{0}^a:Y_t^a)||\mathbf{P}(Y_{0}:Y_t)\big)=0$ for all $t\in \mathbb{Z}_{\geq 0}$, where $KL$ represents the Kullback–Leibler divergence operator.
    \item is $\epsilon$-stealthy if the corresponding observation sequence $Y_{0}^a:Y_t^a$ satisfies
    \begin{equation}\label{ineq:stealthiness}
        KL\big(\mathbf{Q}(Y_{0}^a:Y_t^a)||\mathbf{P}(Y_{0}:Y_t)\big)\leq \log(\frac{1}{1-\epsilon^2}).
    \end{equation}
\end{itemize}
\end{theorem}

\begin{remark}
The $\epsilon$-stealthiness condition defined in~\cite{bai2017data,bai2017kalman} requires $$\lim_{t\to \infty}\frac{KL\big(\mathbf{Q}(Y_{0}^a:Y_t^a)||\mathbf{P}(Y_{0}:Y_t)\big)}{t}\leq \epsilon.$$ 
This allows for the KL divergence to linearly increase over time for any $\epsilon>0$, and as a result, after large-enough time period the attack may be detected. On the other hand, our definition of $\epsilon$-stealthy 
only depends on $\epsilon$ and is fixed for any time $t$; thus, it introduces a stronger notion of stealthiness for the attack. 
\end{remark}
  
\subsection{Formalizing Attack Goal}\label{sec:attack_goal}

As previously discussed, the attacker intends to
\emph{maximize} degradation of control performance. 
Specifically, as 
we consider the origin as the operating point, we formalize 
the attack objective as \emph{maximizing  (the norm of) the states $x_t$}; i.e., moving the system's states into an unsafe region. Since there might be a zone between the safe and unsafe region, we define the the unsafe region as $\mathbf{U}=\{x\in \mathbb{R}^n\,\,|\,\,\Vert x\Vert_2 \geq \alpha\}$ for some $\alpha>R_{\mathbf{S}}$, where $R_{\mathbf{S}}$ is the radius of the safe region $\mathbf{S}$. 
Moreover, the attacker wants \emph{to remain stealthy (i.e., undetected by the intrusion detector)},
as formalized below.

\begin{definition}
\label{def:eps_alpha}
The attack sequence,
denoted by $\{a_{0}, a_{1},...\}$ 
is referred to as $(\epsilon,\alpha)$-successful attack if there exists $t'\in \mathbb{Z}_{\geq 0}$ such that $ \Vert x_{t'}^a \Vert \geq \alpha$ and the 
attack is $\epsilon$-stealthy 
for all $t\in \mathbb{Z}_{\geq 0}$.
When such a sequence exists for a system, the system is called $(\epsilon,\alpha)$-attackable. 
When the system 
is $(\epsilon,\alpha)$-attackable for arbitrarily large $\alpha$, the system is referred to as a perfectly attackable system.
\end{definition}


Now, the problem considered in this work can be formalized as 
capturing the potential impact of stealthy attacks on a considered system; specifically, in the next section,
we derive conditions for 
existence of a \emph{stealthy} yet \emph{effective} attack sequence
$a_{0}, a_{1},...$  resulting in
$\Vert x_t^a\Vert \geq \alpha$ for some $t\in \mathbb{Z}_{\geq 0}$ -- i.e., we find conditions for the system to be $(\epsilon,\alpha)$-attackable. 
Here, for an attack to be stealthy, we focus on 
the $\epsilon-$stealthy notion; 
i.e., that even the best 
intrusion detector could only improve the detection probability by $\epsilon$ compared to the random-guess baseline detector.

\section{Vulnerability Analysis of Nonlinear Systems to Stealthy Attacks} \label{sec:perfect}

In this section, we derive the conditions such that the nonlinear system~\eqref{eq:plant} with closed-loop dynamics~\eqref{eq:closed-loop} is vulnerable to effective stealthy attacks formally defined in Section~\ref{sec:stealthy}. The following theorem captures such condition.

\begin{theorem}
\label{thm:PAt}
The system~\eqref{eq:plant} is ($\epsilon,\alpha$)-attackable 
for arbitrarily large $\alpha$ and arbitrarily small $\epsilon$, if the closed-loop system~\eqref{eq:closed-loop} is incrementally exponentially stable 
(IES) in the set $\mathbf{S}$ and the system~\eqref{eq:input-state} is incrementally unstable 
(IU) in the set $\mathbf{S}$. 
\end{theorem}

\begin{proof}
Assume that the trajectory of the system and controller states for $t\in \mathbb{Z}_{<0}$ is denoted by {$\mathbf{X}_{-\infty}:\mathbf{X}_{-1}$}.
Following attack start at $t=0$, let us consider 
the evolutions of the system with and without attacks during $t\in \mathbb{Z}_{\geq 0}$. 
For the system under attack, starting at time zero, the  trajectory {$\mathbf{X}_{0}^a:\mathbf{X}_{t}^a$} 
of the system and controller states is governed by
\begin{equation}\label{eq:attack_trajec}
\begin{split}
x_{t+1}^a=&f(x_t^a,u_t^a)+ w_t,\quad y_t^{c,a}=h(x_t^a)+v_t+a_t\\
\mathpzc{X}_{t}^a=&f_c(\mathpzc{X}^a_{t-1},y_t^{c,a}),\quad u_t^a=h_c(\mathpzc{X}^a_{t},y_t^{c,a}).\\
\end{split}
\end{equation}

On the other hand, if the system were not under attack 
during $t\in \mathbb{Z}_{\geq 0}$,  we denote the plant and controller state evolution by {$\mathbf{X}_{0}:\mathbf{X}_{t}$}. Hence, it is a continuation of the system trajectories {$\mathbf{X}_{-\infty}:\mathbf{X}_{-1}$} if hypothetically no data-injection attack occurs during $t\in \mathbb{Z}_{\geq 0}$. Since the system and measurement noises are independent of the state, we can assume that $w_t^a=w_t$ and $v_t^{a}=v_t$. In this case, the dynamics of the plant and controller state evolution satisfies
\begin{equation}\label{eq:free_trajec}
\begin{split}
x_{t+1}=&f(x_t,u_t)+ w_t,\quad y_t^c=h(x_t)+v_t,\\
\mathpzc{X}_{t}=&f_c(\mathpzc{X}_{t-1},y_t^c),\quad u_t=h_c(\mathpzc{X}_{t},y_t^c),\\
\end{split}
\end{equation}
which can be captured in the compact form~\eqref{eq:closed-loop}, 
with 
$\mathbf{X}_{0}=\begin{bmatrix}x_{0}\\\mathpzc{X}_{0}\end{bmatrix}$. 

Now, consider the sequence of attack vectors injected in the system from~\eqref{eq:attack_trajec}, which are constructed using the following dynamical model
\begin{equation}\label{eq:attack_seq}
\begin{split}
s_{t+1}&=f(x_t^a,u_t^a) -   f(x_t^a-s_t,u_t^a) \\
a_t&=h(x_t^a-s_t)-h(x_t^a),
\end{split}
\end{equation}
for $t\in \mathbb{Z}_{\geq 0}$, and with some arbitrarily chosen nonzero initial value of $s_0$. By injecting the above attack sequence into the sensor measurements, we can verify that $y_t^{c,a}=h(x_t^a)+v_t+a_t=h(x_t^a-s_t)+v_t$. After defining 
\begin{equation}
e_t\delequal x_t^a-s_t,
\end{equation}
and combining~\eqref{eq:attack_seq} with~\eqref{eq:attack_trajec}, the dynamics of $e_t$ and the controller, and the corresponding input and output satisfy
\begin{equation}\label{eq:closed_loop_attack}
\begin{split}
e_{t+1}=&f(e_t,u_t^a)+ w_t,\quad y_t^{c,a}=h(e_t)+v_t,\\
\mathpzc{X}_{t}^a=&f_c(\mathpzc{X}^a_{t-1},y_t^{c,a}), \quad u_t^a=h_c(\mathpzc{X}^a_{t},y_t^{c,a}),
\end{split}
\end{equation}
with the initial condition $e_0=x_0^a-s_0$. 

Now, if we define $\mathbf{X}^e_t=\begin{bmatrix}e_{t}\\\mathpzc{X}_{t}^a\end{bmatrix}$, it holds that
\begin{equation}\label{eq:closed_attack}
\mathbf{X}^e_{t+1}=F(\mathbf{X}^e_{t},\mathbf{W}_t).
\end{equation}
with $\mathbf{X}^e_{0}=\begin{bmatrix}e_{0}\\\mathpzc{X}_{0}^a\end{bmatrix}$. Since we have that $x_0^a=x_0$ and $\mathpzc{X}_{0}^a=\mathpzc{X}_{0}$, it holds that $\mathbf{X}_{0}-\mathbf{X}^e_{0}=\begin{bmatrix}s_{0}\\0\end{bmatrix}$. On the other hand, since  both~\eqref{eq:closed_attack} and~\eqref{eq:closed-loop} share the same function and argument $\mathbf{W}_t$, 
the closed-loop system~\eqref{eq:closed_attack} is 
IES, and it also follows that
\begin{equation}\label{eq:error_bound}
\begin{split}
\Vert \mathbf{X}(t,\mathbf{X}_{0},\mathbf{W})-\mathbf{X}^e(t,\mathbf{X}^e_{0},\mathbf{W})\Vert &\leq \kappa \Vert \mathbf{X}_{0}-\mathbf{X}^e_{0}\Vert \lambda^{-t}\\
&\leq \kappa \Vert s_0\Vert \lambda^{-t};
\end{split}
\end{equation}
therefore, the trajectories of $\mathbf{X}$ (i.e., the system without attack)  and $\mathbf{X}^e$ converge to each other exponentially fast. 

We now use these results to show that the generated attack sequence satisfies the $\epsilon$-stealthiness condition. By defining $\mathbf{Z}_{t}=\begin{bmatrix}x_{t}\\y_t^{c}\end{bmatrix}$ and $\mathbf{Z}_{t}^e=\begin{bmatrix}e_{t}\\y_t^{c,a}\end{bmatrix}$, it holds that
\begin{equation}\label{ineq:1}
\begin{split}
&KL\big(\mathbf{Q}(Y_{0}^a:Y_t^a)||\mathbf{P}(Y_{0}:Y_t)\big) \\
&\stackrel{(i)}\leq KL\big(\mathbf{Q}(\mathbf{X}_{0}^e:\mathbf{X}_{t}^e)||\mathbf{P}(\mathbf{X}_{0}:\mathbf{X}_{t})\big) \\
&\stackrel{(ii)}\leq KL\big(\mathbf{Q}(\mathbf{Z}_{-\infty}:\mathbf{Z}_{-1},\mathbf{Z}_{0}^e:\mathbf{Z}_{t}^e)||\mathbf{P}(\mathbf{Z}_{-\infty}:\mathbf{Z}_{-1},\mathbf{Z}_{0}:\mathbf{Z}_{t})\big),
\end{split}
\end{equation}
where we applied the data-processing inequality property of KL-divergence for $t\in \mathbb{Z}_{\geq 0}$ to obtain $(i)$, and the monotonicity property of KL-divergence to obtain the inequality $(ii)$.\footnote{Due to the space limitation, we do not introduce data-processing, chain-rule, and monotonicity properties of KL-divergence. More information about these terms can be found in~\cite{thomas2006elements}.}  Then, we apply the chain-rule property of KL-divergence on the right-hand side of 
\eqref{ineq:1} to obtain the following
\begin{equation}\label{ineq:2}
\begin{split}
&KL\big(\mathbf{Q}(\mathbf{Z}_{-\infty}:\mathbf{Z}_{-1},\mathbf{Z}_{0}^e:\mathbf{Z}_{t}^e)||\mathbf{P}(\mathbf{Z}_{-\infty}:\mathbf{Z}_{-1},\mathbf{Z}_{0}:\mathbf{Z}_{t})\big)\\
&= KL\big(\mathbf{Q}(\mathbf{Z}_{-\infty}:\mathbf{Z}_{-1})||\mathbf{P}(\mathbf{Z}_{-\infty}:\mathbf{Z}_{-1})\big)+\\
&\,\,\,\,\,\,\,KL\big(\mathbf{Q}(\mathbf{Z}_{0}^e:\mathbf{Z}_{t}^e|\mathbf{Z}_{-\infty}:\mathbf{Z}_{-1})||\mathbf{P}(\mathbf{Z}_{0}:\mathbf{Z}_{t}|\mathbf{Z}_{-\infty}:\mathbf{Z}_{-1})\big)\\
&=KL\big(\mathbf{Q}(\mathbf{Z}_{0}^e:\mathbf{Z}_{t}^e|\mathbf{Z}_{-\infty}:\mathbf{Z}_{-1})||\mathbf{P}(\mathbf{Z}_{0}:\mathbf{Z}_{t}|\mathbf{Z}_{-\infty}:\mathbf{Z}_{-1})\big);
\end{split}
\end{equation}
here, we used the fact that the 
KL-divergence of two identical distributions (i.e., $\mathbf{Q}(\mathbf{Z}_{-\infty}:\mathbf{Z}_{-1})$ and $\mathbf{P}(\mathbf{Z}_{-\infty}:\mathbf{Z}_{-1})$ since the system is not under attack for $t<0$) 
is zero.

Applying the chain-rule property of 
KL-divergence to~\eqref{ineq:2} results in
\begin{equation}\label{ineq:3}
\begin{split}
K&L\big(\mathbf{Q}(\mathbf{Z}_{0}^e:\mathbf{Z}_{t}^e|\mathbf{Z}_{-\infty}:\mathbf{Z}_{-1})||\mathbf{P}(\mathbf{Z}_{0}:\mathbf{Z}_{t}|\mathbf{Z}_{-\infty}:\mathbf{Z}_{-1})\big)\\
&\leq KL\big(\mathbf{Q}(e_0|\mathbf{Z}_{-\infty}:\mathbf{Z}_{-1})||\mathbf{P}(x_{0}|\mathbf{Z}_{-\infty}:\mathbf{Z}_{-1})\big)\\
&+ KL\big(\mathbf{Q}(y_0^{c,a}|e_0,\mathbf{Z}_{-\infty}:\mathbf{Z}_{-1})||\mathbf{P}(y_{0}|x_0,\mathbf{Z}_{-\infty}:\mathbf{Z}_{-1})\big)\\
&+...+ KL\big(\mathbf{Q}(e_t|\mathbf{Z}_{-\infty}:\mathbf{Z}_{t-1}^e)||\mathbf{P}(x_{t}|\mathbf{Z}_{-\infty}:\mathbf{Z}_{t-1})\big)\\
&+ KL\big(\mathbf{Q}(y_t^{c,a}|e_t,\mathbf{Z}_{-\infty}:\mathbf{Z}^e_{t-1})||\mathbf{P}(y_{t}|x_t,\mathbf{Z}_{-\infty}:\mathbf{Z}_{t-1})\big).
\end{split}
\end{equation}

Given $\mathbf{Z}_{-\infty}:\mathbf{Z}_{t-1}$, the distribution of $x_t$ is a Gaussian with mean $f(x_{t-1},u_{t-1})$ and covariance $\Sigma_w$. Similarly given  $\mathbf{Z}_{-\infty}:\mathbf{Z}_{-1},\mathbf{Z}_{0}^e:\mathbf{Z}_{t-1}^e$,  the distribution of $e_t$ is a Gaussian with mean $f(e_{t-1},u_{t-1}^a)$ and covariance $\Sigma_w$. Since we have that $x_t=f(x_{t-1},u_{t-1})+w_t$ and $e_t=f(e_{t-1},u_{t-1}^a)+w_t$ according to~\eqref{eq:free_trajec} and~\eqref{eq:closed_loop_attack}, it holds that $f(x_{t-1},u_{t-1})-f(e_{t-1},u_{t-1}^a)=x_t-e_t$. On the other hand, in~\eqref{eq:error_bound} we showed that $\Vert x_t-e_t\Vert \leq \kappa \Vert s_0\Vert \lambda^{-t}$ holds for $t\in \mathbb{Z}_{\geq 0}$. Therefore, for all $t\in \mathbb{Z}_{\geq 0}$, it holds that
\begin{equation}\label{ineq:4}
\begin{split}
KL\big(\mathbf{Q}(e_t|\mathbf{Z}_{-\infty}:&\mathbf{Z}_{t-1}^e)||\mathbf{P}(x_{t}|\mathbf{Z}_{-\infty}:\mathbf{Z}_{t-1})\big)= \\
=& (x_t-e_t)^T \Sigma_w^{-1}  (x_t-e_t) 
\\\leq& \kappa^2 \Vert s_0\Vert^2 \lambda^{-2t} \lambda_{max}(\Sigma_w^{-1}),
\end{split}
\end{equation}
where $\lambda_{max}(\Sigma_w^{-1})$ is the maximum eigenvalue of the matrix~$\Sigma_w^{-1}$. 

Now, using the Markov property it holds that $\mathbf{Q}(y_t^{c,a}|e_t,\mathbf{Z}_{-\infty}:\mathbf{Z}^e_{t-1})=\mathbf{Q}(y_t^{c,a}|e_t)$ and $\mathbf{P}(y_{t}|x_t,\mathbf{Z}_{-\infty}:\mathbf{Z}_{t-1})=\mathbf{P}(y_{t}|x_t)$; 
also, from~\eqref{eq:free_trajec} and~\eqref{eq:closed_loop_attack} it holds that given $x_t$ and $e_t$, $\mathbf{P}(y_{t}|x_t)$ and $\mathbf{Q}(y_t^{c,a}|e_t)$ are both Gaussian with mean $h(x_t)$ and $h(e_t)$, respectively and covariance $\Sigma_v$. Thus, it follows that
\begin{equation}\label{ineq:5}
\begin{split}
KL\big(\mathbf{Q}&(y_t^{c,a}|e_t)||\mathbf{P}(y_{t}|x_t)\big)\\
&=\big(h(x_t)-h(e_t)\big)^T \Sigma_v^{-1}  \big(h(x_t)-h(e_t)\big)\\
&\leq L_h^2 (x_t-e_t)^T \Sigma_v^{-1}  (x_t-e_t)\\
&\leq L_h^2 \kappa^2 \Vert s_0\Vert^2 \lambda^{-2t} \lambda_{max}(\Sigma_v^{-1}).
\end{split}
\end{equation}

Combining~\eqref{ineq:1}-\eqref{ineq:5} results in 
\begin{equation}\label{eq:epsilon}
\begin{split}
&KL\big(\mathbf{Q}(Y_{0}^a:Y_t^a)||\mathbf{P}(Y_{0}:Y_t)\big) \leq \\
&\sum_{i=0}^t \kappa^2 \Vert s_0\Vert^2 \lambda^{-2t} \lambda_{max}(\Sigma_w^{-1})+L_h^2 \kappa^2 \Vert s_0\Vert^2 \lambda^{-2t} \lambda_{max}(\Sigma_v^{-1})\\
&\leq \frac{\kappa^2 \Vert s_0\Vert^2}{1-\lambda^2} \big(\lambda_{max}(\Sigma_w^{-1})+L_h^2 \lambda_{max}(\Sigma_v^{-1})\big)\delequal b_{\epsilon}.
\end{split}
\end{equation}
Finally, with $b_{\epsilon}$ defined as in~\eqref{eq:epsilon}, the attack sequence defined in~\eqref{eq:attack_seq} satisfies the $\epsilon$-stealthiness condition with $\epsilon=\sqrt{1-e^{-b_{\epsilon}}}$. 

We now show that the proposed attack sequence is effective; i.e., there exists $t'\in \mathbb{Z}_{\geq 0}$ such that $\Vert x_{t'}^a\Vert \geq \alpha$ for arbitrarily large $\alpha$. To achieve this, consider the two dynamics from~\eqref{eq:attack_trajec} and~\eqref{eq:closed_loop_attack} for any $t\in \mathbb{Z}_{\geq 0}$
\begin{equation}
\begin{split}
x_{t+1}^a=&f(x_t^a,u_t^a)+ w_t=f_u(x_t^a,U_t^a)\\
e_{t+1}=&f(e_t,u_t^a)+ w_t=f_u(e_t,U_t^a)
\end{split}
\end{equation}
with $U_t^a=\begin{bmatrix}{u_t^a}^T&w_t^T\end{bmatrix}^T$, for $t\in \mathbb{Z}_{\geq 0}$. Since we assumed that the open-loop system~\eqref{eq:input-state} is 
IU on the set $\mathbf{S}$, it holds that for all $x_0^a=x_0\in \mathbf{S}$, there exits a nonzero $s_0$ such that for any $M>0$ 
\begin{equation}\label{eq:alpha}
\Vert x^a(t,x_0^a,U^a)-e(t,x_0^a-s_0,U^a)\Vert \geq M
\end{equation}
holds in $t\geq t'$, for some $t'\in \mathbb{Z}_{\geq 0}$. 

On the other hand, we showed in~\eqref{eq:error_bound} that $\Vert x(t,x_0,U)-e(t,x_0^a-s_0,U^a)\Vert \leq \kappa \Vert s_0\Vert \lambda^{-t}$. Combining this with~\eqref{eq:alpha} and using the fact that $\Vert x(t,x_0,U)\Vert \leq R_{\mathbf{S}}$ results in
\begin{equation}\label{eq:alpha_x_a}
\begin{split}
&\Vert x^a(t,x_0^a,U^a)-x(t,x_0-s_0,U)\Vert = \\
&\Vert x^a(t,x_0^a,U^a)-e(t,x_0^a-s_0,U^a)+e(t,x_0^a-s_0,U^a)\\&-x(t,x_0-s_0,U)\Vert \geq \Vert x^a(t,x_0^a,U^a)-e(t,x_0^a-s_0,U^a)\Vert \\
& -\Vert e(t,x_0^a-s_0,U^a)-x(t,x_0-s_0,U)\Vert 
\geq M-\kappa \Vert s_0\Vert \lambda^{-t} \\
&\Rightarrow \Vert x^a(t,x_0^a,U^a)\Vert \geq M-\kappa \Vert s_0\Vert \lambda^{-t}-R_{\mathbf{S}}\\
&\qquad \qquad \qquad \qquad  \geq  M-\kappa \Vert s_0\Vert -R_{\mathbf{S}}.
\end{split}
\end{equation}

Since $M$ is arbitrarily, we can choose it to satisfy $M>\alpha+R_s+\kappa \Vert s_0\Vert$, for arbitrarily large $\alpha$. Thus, the system is $(\epsilon,\alpha)$-attackable.
\end{proof}

From~\eqref{eq:closed_loop_attack}, we can see that the false sensor measurements are generated by the evolution of $e_t$. Therefore, intuitively, the attacker wants to fool the system into believing that $e_t$ is the actual state of the system instead of $x_t^a$. Since $e_t$ and $x_t$ (i.e., the system state if no attack occurs during $t\in \mathbb{Z}_{\geq 0}$) converge to each other exponentially fast, the idea is that the system almost believes that $x_t$ is the system state (under attack), while the actual state $x_t^a$ becomes arbitrarily large.

Furthermore,
all parameters $\kappa$, $ \lambda$, $L_h$, $\Sigma_w$, and $\Sigma_v$ in~\eqref{eq:epsilon} are some constants that depend either on system properties ($L_h$, $\Sigma_w$, and $\Sigma_v$) or are determined by the controller design ($\kappa$, $ \lambda$). However, $s_0$ is set by the attacker, and \emph{it can be chosen arbitrarily small to make $\epsilon$ arbitrarily close to zero}.   Yet, 
$s_0$ can not be equal to zero; in that case~\eqref{eq:alpha} would not hold --  i.e., the attack would not not be impactful. Therefore, as opposed to attack methods targeting the prediction covariance in~\cite{bai2017kalman} where the attack impact linearly changes with $\epsilon$, here arbitrarily large $\alpha$ (high impact attacks) can be achieved  even with an arbitrarily small $\epsilon$ -- it may only take more time to get to $\Vert x^a_{t'}\Vert \geq \alpha$.

\begin{remark}
Even though we assumed that the closed-loop dynamics is 
IES, slightly weaker results can still be obtained for closed-loop dynamics with incrementally asymptotic stability. 
We will consider this case as future work. 
\end{remark}

\begin{remark}
For constructing the attack sequence in~\eqref{eq:attack_seq} we assumed that the attacker has knowledge of the system's nonlinear functions $f$ and $h$, as well as has access to the values of the system state. 
In future work, we will show how these assumptions 
can be relaxed for systems with general nonlinear dynamics.
\end{remark}

\begin{remark}
In case that either $w=0$ or $v=0$ (i.e., when there is no process or measurement noise), one can still get a similar bound on the KL-divergence only as a function of the nonzero noise covariance by applying monotonicity and data-processing inequality. However, ensuring stealthiness requirement is not possible if both $w=0$ and $v=0$ (i.e., for the noiseless system), as the system would be completely  deterministic, and thus theoretically any small perturbation to the sensor measurements could be detected.
\end{remark}

\subsection{Vulnerability Analysis of LTI Systems} 

Theorem~\ref{thm:PAt} can also be applied to find the condition for the existence of ($\epsilon,\alpha$)-successful attacks on LTI systems. Specifically, the LTI 
formulation of~\eqref{eq:plant} and~\eqref{eq:plant_withoutB} is 
\begin{equation}\label{eq:lti}
\begin{split}
{x}_{t+1} &= Ax_t+Bu_t+w_t,\quad y_t^c = Cx_t+v_t,\\
\mathpzc{X}_{t} &= A_c\mathpzc{X}_{t-1}+B_cy_t^c,\quad u_t  = C_c\mathpzc{X}_{t};\\
\end{split}    
\end{equation}
LTI systems with any controller (e.g., LQG controllers) can be captured in the above form. 
The following lemma provides the conditions for IES and IU for the above 
LTI~system.

\begin{lemma}\label{lemma:LTI}
Consider the LTI dynamical system in the form of $x_{t+1}=Ax_t+Bd_t$. The system is IES if and only if all eigenvalues of the matrix $A$ are inside the unit circle. The system is 
IU if and only if $A$ has an unstable eigenvalue. 
\end{lemma}
\begin{proof}
The proof is straightforward and 
follows from the definition and the direct method of Lyapunov. 
\end{proof}

This allows us to directly capture conditions for stealthy yet effective attacks on LTI systems. 
\begin{corollary}
\label{cor:cor1}
The LTI system~\eqref{eq:lti} is ($\epsilon,\alpha$)-attackable for arbitrarily large $\alpha$ if the matrix $A$ is unstable and the closed-loop control system is asymptotically stable. 
\end{corollary}

\begin{proof}
The proof is directly obtained\textbf{} by combining Theorem~\ref{thm:PAt} and Lemma~\ref{lemma:LTI}. 
\end{proof}

Asymptotic stability of the closed-loop system is not a restrictive assumption as stability is commonly the weakest required performance guarantee for a control system. Matrix $A$ being unstable is a necessary and sufficient condition for satisfying ($\epsilon,\alpha$)-attackability when any set of sensors can be compromised. Note that the  ($\epsilon,\alpha$)-attackability condition for LTI systems with an optimal detector 
complies with the results from~\cite{mo2010false,jovanov_tac19} where LQG controllers with residue based detectors (e.g., $\chi^2$ detectors) have been considered.   

\begin{remark}
The false-date injection attack sequence design method from~\eqref{eq:attack_seq} will reduce into a simple dynamical model
\begin{equation}
\begin{split}
s_{t+1}&=Ax_t^a+Bu_t^a - (A(x_t^a-s_t)+Bu_t^a)=As_t \\
a_t&=C(x_t^a-s_t)-C(x_t^a)=-Cs_t,
\end{split}
\end{equation}
that only requires knowledge about the matrices $A$ and $C$. In addition, unlike the case for nonlinear systems, there is no need to have access to the actual states of the system.  
\end{remark}

\begin{remark}
In~\cref{sec:attack_model} we assumed that $\mathcal{K}=\mathcal{S}$; i.e., the attacker can compromise all sensors. However, when the system is LTI, the minimum subset of compromised sensors can be obtained as
\begin{equation}
\min_{v_i\in \{v_1,...,v_q\}} \Vert \text{supp} (Cv_i)\Vert_0, 
\end{equation}
where $ \{v_1,...,v_q\}$ denotes the set of unstable eigenvectors of the matrix $A$, and $\text{supp}$ denotes the set of nonzero elements of the vector. 
\end{remark}

\section{Simulation Results}
\label{sec:simulation}
We illustrate 
our results 
on a 
case-study. 
%
Specifically, we consider a fixed-base inverted pendulum equipped with an Extended Kalman Filter to estimate the states of the system followed by a feedback full state controller to keep the pendulum rod in the inverted position. 
Using $x_1=\theta$ and $x_2=\dot{\theta}$, 
the inverted pendulum dynamics can be modeled~as
\begin{equation} \label{eq:state_Nonlin}
\begin{split}
\dot{x}_1&=x_2\\
\dot{x}_2&=\frac{g}{r}\sin{x_1} - \frac{b}{mr^2} x_2+ \frac{L}{mr^2};
\end{split}
\end{equation}
here, $\theta$ is the angle of
pendulum rod from the vertical axis measured clockwise, $b$ is the Viscous friction coefficient, $r$ is the radius of inertia of the pendulum about the fixed point, 
$m$ is the mass of the pendulum, $g$ is the acceleration due to gravity, and $L$ is the external torque that is applied at the fixed base. We assumed that both the states are measured by sensors.
Finally, we assumed $g=9.8$, $m=.2 Kg$, $b=.1$, $r=.3m$, $\Sigma_w=\Sigma_v=\begin{bmatrix}.01&0\\0&.01\end{bmatrix}$ and discretized the model with $T_s=10~ms$. 
We assume the safe region for angle around the equilibrium point $\theta=0$ is $\mathbf{S}=(-\frac{\pi}{3},\frac{\pi}{3})$. To detect the presence of attack, we designed a standard $\chi^2$-based anomaly detector that receives the sensor values and outputs the residue/anomaly alarm.    
%

We used the attack model considered in~\eqref{eq:attack_seq} to generate the sequence of false-data injection attacks over time. Fig.~\ref{fig:attack}(a) shows the angle of the pendulum pod over time. Before the attack starts at time zero, the pendulum pod is around the angle zero; however, after initiating the attack it can be observed that the absolute value of the angle increases over time until it leaves the safe set and even becomes more than $\pi$. Note that having values more than $\pi$ does not make a difference because we have a periodic system, {and $\pi$ corresponds to the pendulum falling down}.   Meanwhile, the distribution of the norm of the residue signal (see Fig.~\ref{fig:attack}(b)) does not change before and after attack initiation -- i.e., the attack remains~stealthy.



\begin{figure}[!t]
\centering
{\label{fig:theta}
  \includegraphics[clip,width=0.484\columnwidth]{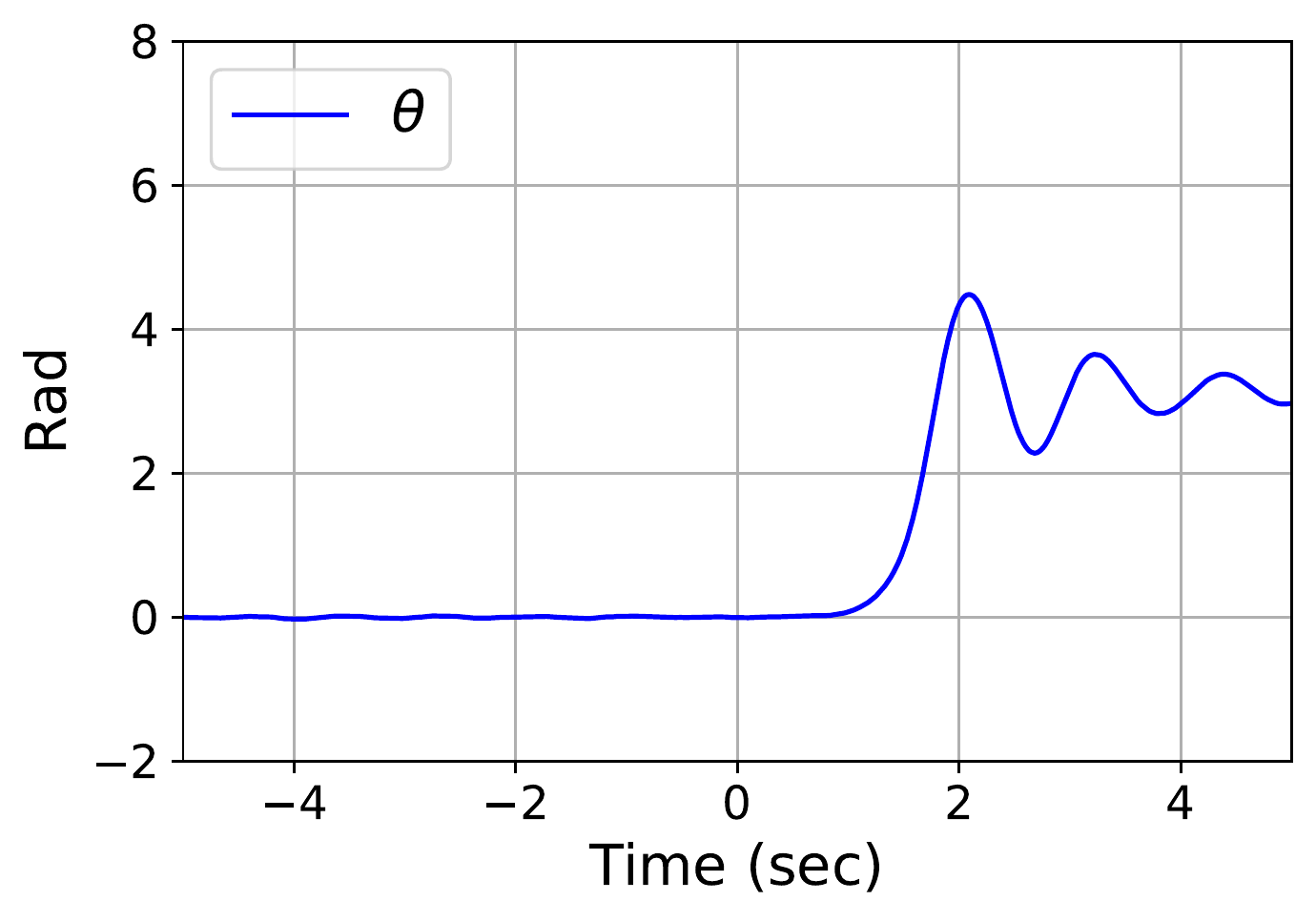}
}
{\label{fig:residue}
  \includegraphics[clip,width=0.484\columnwidth]{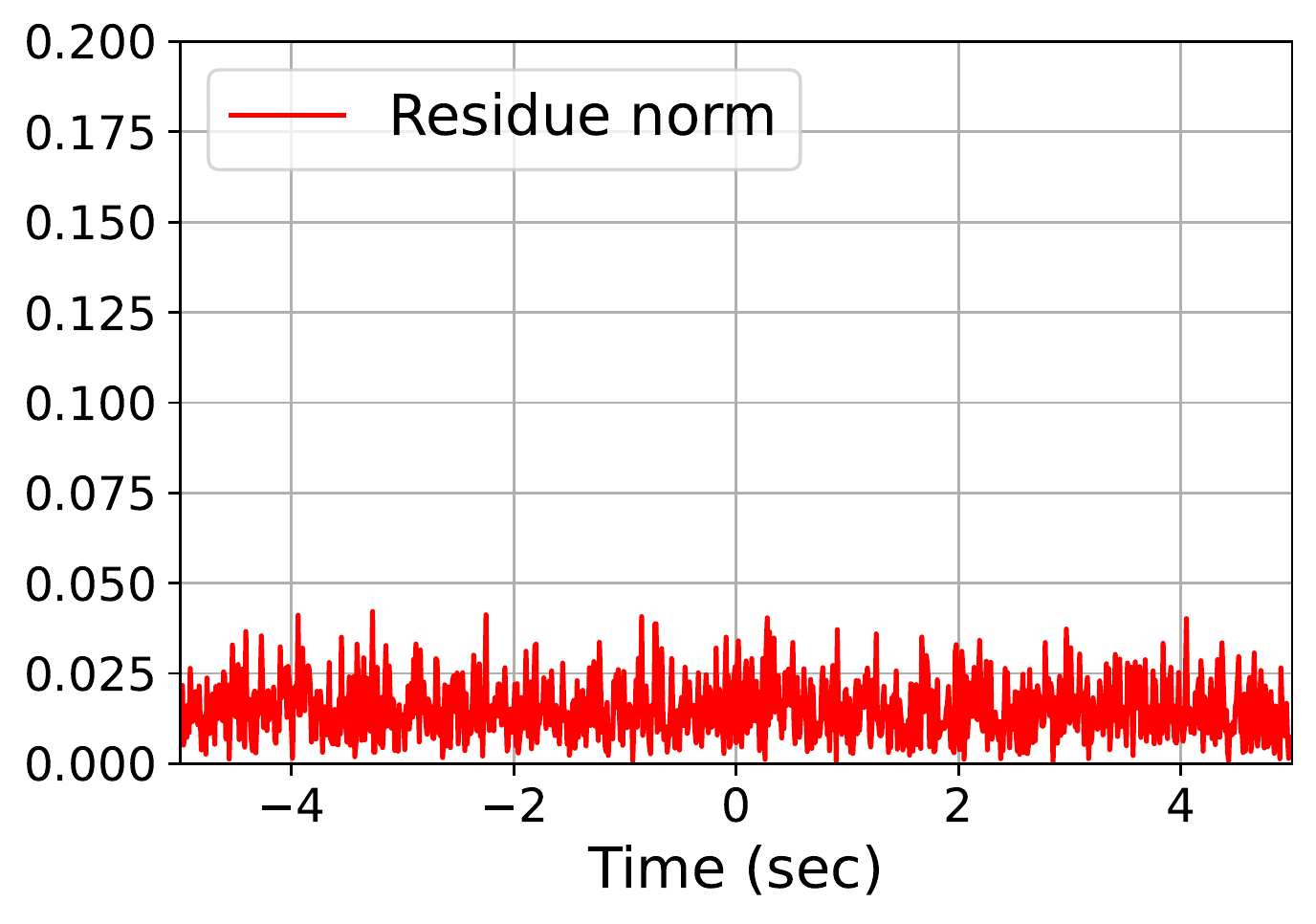}%
}

\caption{(a) Angle's ($\theta$) absolute value over time for the under-attack system, when the attack starts at time zero; (b) The residue norm over time for the under-attack system, when the attack starts at time zero. }\label{fig:attack}
\end{figure}

\section{Conclusion and Future Work}
\label{sec:concl}
In this paper, we have considered the problem of vulnerability analysis for nonlinear control systems with Gaussian noise, when attacker can compromise sensor measurements from any subset of sensors. Notions of strict stealthiness and $\epsilon$-stealthiness have been defined, and we have shown that these notions are independent of the deployed intrusion detector. Using the KL-divergence, 
we have presented conditions for the existence of stealthy yet effective attacks. Specifically, we have defined the $(\epsilon,\alpha)$-successful attacks where the goal of the attacker is to be $\epsilon$-stealthy while moving the system states into an unsafe region, determined by the parameter $\alpha$. We have then derived a condition for which there exists a sequence of such $(\epsilon,\alpha)$-successful false-data injection  attacks. In particular, we showed that if the closed-loop system is incrementally exponentially stable and the open-loop system is incrementally unstable, then there exists a sequence of $(\epsilon,\alpha)$-successful attacks. We also provided the results for LTI systems, showing that they are compatible with the existing results for LTI systems and $\chi^2$-based detectors. 

Our results assume that the attacker has knowledge of the state evolution function $f$, as well as access to the values of the actual system states and the control inputs during the attack. Future work will be directed toward deriving conditions when the attacker has limited knowledge about the states, control input and the function $f$. We will also study the effects of specific previously reported  attacks (e.g., replay attack) on general nonlinear control systems using the defined notions of strict and $\epsilon$-stealthiness.

\bibliographystyle{IEEEtranMod}


\end{document}